\documentclass[12pt]{article}
\usepackage{arxiv}

\usepackage[utf8]{inputenc} %
\usepackage[T1]{fontenc}    %
\usepackage{url}            %
\usepackage{booktabs}       %
\usepackage{amsfonts}       %
\usepackage{nicefrac}       %
\usepackage{microtype}      %
\usepackage{xcolor}         %
\usepackage[inline]{enumitem} %
\usepackage[numbers,sort&compress]{natbib}
\usepackage{doi}
\usepackage{algorithm,algpseudocode}
\usepackage{amsmath}
\usepackage{amssymb}
\usepackage{mathtools}
\usepackage{amsthm}
\usepackage{mathrsfs}
\usepackage{multirow, array, threeparttable}
\usepackage[font=normalsize,skip=4pt]{caption}
\usepackage{subcaption}
\usepackage{graphicx}   %
\hypersetup{
  colorlinks,
  linkcolor={blue!60!black},
  citecolor={red!60!black},
  urlcolor={blue!80!black}
}
\usepackage[capitalize,nameinlink]{cleveref}

\makeatletter
\renewcommand{\paragraph}{%
  \@startsection{paragraph}{4}{\z@}%
  {0.25\baselineskip}  %
  {-0.25em}            %
  {\normalfont\normalsize\bfseries}%
}
\makeatother

\usepackage[all, cmtip]{xy}
\usepackage{tikz-cd}
\usepackage[strings]{underscore}
\usepackage[strict=true]{csquotes}
\usepackage{siunitx}
\newtheorem{theorem}{Theorem}

\algnewcommand\algorithmicinput{\textbf{Input:}}
\algnewcommand\algorithmicoutput{\textbf{Output:}}
\algnewcommand\algorithmicnote{\textbf{Note:}}
\algnewcommand\Input{\item[\algorithmicinput]}%
\algnewcommand\Output{\item[\algorithmicoutput]}%
\algnewcommand\Note{\item[\algorithmicnote]}%

\DeclareMathOperator{\diag}{diag} 
\DeclareMathOperator{\sym}{\text{sym}}
\DeclareMathOperator{\tr}{\text{tr}}

\title{Stochastic Subspace via Probabilistic Principal Component Analysis
  for Characterizing Model Error}

\author{{\hspace{1mm}Akash Yadav} \\
  University of Houston\\
  \texttt{ayadav4@uh.edu} \\
  \And
  {\hspace{1mm}Ruda Zhang} \\
  University of Houston\\
  \texttt{rudaz@uh.edu} \\
}

\date{}

\renewcommand{\headeright}{ }
\renewcommand{\shorttitle}{SS-PPCA}

\begin{document}
\maketitle

  \begin{abstract}
    This paper proposes a probabilistic model of subspaces based on
    the probabilistic principal component analysis (PCA).
    Given a sample of vectors in the embedding space---commonly known as a snapshot matrix---this method
    uses quantities derived from the probabilistic PCA to construct distributions of
    the sample matrix, as well as the principal subspaces.
    It is applicable to projection-based reduced-order modeling methods,
    such as proper orthogonal decomposition and related model reduction methods.
    The stochastic subspace thus constructed can be used, for example,
    to characterize model-form uncertainty in computational mechanics.
    The proposed method has multiple desirable properties:
    (1) it is naturally justified by the probabilistic PCA
    and has analytic forms for the induced random matrix models;
    (2) it satisfies linear constraints, such as boundary conditions of all kinds, by default;
    (3) it has only one hyperparameter, which significantly simplifies training; and
    (4) its algorithm is very easy to implement.
    We demonstrate the performance of the proposed method  
    via several numerical examples in computational mechanics and structural dynamics.
\end{abstract}

\section{Introduction}

Engineering systems can often be described by a set of partial
differential equations, commonly referred to as governing equations.
These governing equations are typically derived by making assumptions
and simplifying the underlying physical process,
which inevitably introduces errors \cite{Oberkampf2002,Roy2011}.
The errors are further manifested via numerical approximations, imprecise initial and boundary conditions,
unknown model parameters,
and errors in the differential equations themselves.
The discrepancy between model predictions and ground truth is referred to as model error.
Model error (also known as model discrepancy \cite{Kennedy2001}, model inadequacy \cite{Kennedy2001},
structural uncertainty \cite{Trucano2006}, or model-form error \cite{Roy2011} in various contexts) is ubiquitous in computational engineering,
but its probabilistic analysis has been a notoriously difficult challenge.
Unlike parametric uncertainties that are associated with model parameters,
model-form uncertainties concern the variability of the model itself,
which are inherently nonparametric.
\citet{Pernot2017} argues that model-form uncertainty should not be absorbed in parametric uncertainty as it leads to prediction bias.
The analysis of model-form uncertainty can be divided into two sub-problems:
1) characterization and 2) correction.
Determining the predictive error of an inadequate model
is categorized as \textit{characterization}.
Reducing the predictive error of an inadequate model via some form of adjustment
is called \textit{correction}, which is much more complex than characterization.
There are various studies in the literature that address model error.
The approaches can be divided into two categories: 1) direct representation and 2) indirect representation.

The methods based on direct representation are mainly subdivided into two categories,
depending on where model errors are addressed: 1) external, and 2) internal \cite{He2016}.
In \textit{external direct representation} methods,
model error is accounted for by adding a correction term to the model output,
which is usually calibrated using a Gaussian process to match the observed data.
One of the earliest studies to address model error
was carried out by \citet{Kennedy2001} in 2001,
in which the authors presented a Bayesian calibration technique
as an improvement on the traditional techniques.
Since then, the Bayesian inferential and modeling framework has been adopted
and further developed by many studies \cite{Higdon2004,Higdon2008,Bayarri2007,Qian2008,Gramacy2008}.
KOH Bayesian calibration corrects the model output
but has limited extrapolative capability to unobserved quantities \cite{Maupin2020}.
\citet{Farrell2015} present an adaptive modeling paradigm for Bayesian calibration, and model selection,
where the authors developed an Occam-Plausibility algorithm to select the best model with calibration.
However, if the initial set of models is invalid, then another set of models is required,
which requires a lot of prior information and modeling experience.
In general, methods using external direct representation
tend to violate physical constraints \cite{Strong2014,Sargsyan2015},
rely heavily on prior information \cite{Brynjarsdottir2014},
fail to elucidate model error and data error, and
struggle with extrapolation and prediction of unobserved quantity \cite{Bayarri2007,Oliver2015,Pernot2017}.
To alleviate some of the drawbacks, \citet{Brynjarsdottir2014} and \citet{Plumlee2017}
show that prior does play an important role in improving performance.
\citet{He2016} proposed a general physics-constrained model correction framework,
ensuring that the corrected model adheres to physical constraints.
Despite all the developments,
methods based on external direct representation have several limitations, as listed above,
mainly because model error is inherently nonparametric.

\textit{Internal direct representation} methods improve a model through intrusive modifications,
rather than by adding an extra term to the output.
Embedded model enrichment \cite{Oliver2015} provides a framework for extrapolative predictions using composite models.
\citet{Morrison2018} address model inadequacy with stochastic operators
introduced as a source in the governing equations while preserving underlying physical constraints.
Similarly, \citet{Portone2022} propose a stochastic operator
for an unclosed dispersion term in an averaged advection--diffusion equation
for model closure.
However, modifying the governing equations of the model brings
challenges in the computational cost and implementation.
Additionally, significant prior knowledge of the system is required
to alter the governing equations and achieve improvements effectively.
\citet{Bhat2017} address model-form error by
embedding dynamic discrepancies within small-scale models,
using Bayesian smoothing splines
(BS-ANOVA) framework.
\citet{Strong2014} introduce a novel way to quantify model inadequacies
by decomposing the model into sub-functions and addressing discrepancies in each sub-function.
The authors show the application of their method in health economics studies.
\citet{Sargsyan2019} improve uncertainty representation and quantification in computational models
by introducing a Bayesian framework with embedded model error representation using polynomial chaos expansion.
In general, internal direct representation methods
require a deep understanding of the system to implement the changes necessary for improvement.
This knowledge is not always available and often demands an extensive study of the system.

An approach based on \textit{indirect representation} is presented by \citet{Sargsyan2015},
where the authors aim to address model error at the source by hyperparameterizing the parameters of assumptions.
However, it is not always easy to pinpoint the source of the errors
especially when errors arise from multiple sources.
Furthermore, this approach does not correct for model error.
Another indirect representation approach to address model error is based on the concept of
stochastic reduced order modeling, originally introduced by Soize in the early 2000s
\cite{Soize2000, Soize2001, Soize2005} and reviewed in \cite{Soize2005npm}.
In their earlier approaches,
the authors build the reduced-order model (ROM) via projection onto a deterministic subspace and
construct stochastic ROMs by randomizing the coefficient matrices using the principle of maximum entropy.
The term ``stochastic reduced order model'' is also used by a method developed around 2010
in a different context \cite{Grigoriu2009, Grigoriu2012, Warner2013},
which is essentially a discrete approximation of a probability distribution
and can be constructed non-intrusively from a computational model.
Other approaches, such as random field finite elements and probabilistic finite element methods \cite{Liu1986a,Liu1986}, have been developed for handling parametric uncertainty in material properties and loads. While they do not explicitly account for model error, extensions incorporating model discrepancy terms could potentially address this limitation.

A recent work by \citet{Soize2017srob} in model-form uncertainty merges ideas
from projection-based reduced-order modeling \cite{Benner2015,ZhangRD2022gps} and random matrix theory \cite{Mehta2004,Edelman2005}.
In the new approach, instead of randomizing coefficient matrices, the reduced-order basis (ROB) is randomized.
The key observation is that since the ROB determines the ROM,
randomizing the ROB also randomizes the ROM,
which can be used for efficient probabilistic analysis of model error.
Using computationally inexpensive ROM instead of a full-order model
makes the approach computationally tractable for uncertainty quantification tasks.
The probabilistic model and estimation procedure proposed in \cite{Soize2017srob}
have been applied to various engineering problems \cite{WangHR2019,Soize2019srob}.
Because this method can face challenges due to a large number of hyperparameters and complex implementation,
later developments have aimed at reducing the number of hyperparameters and simplifying hyperparameter training \cite{Soize2019srob,Azzi2022}.
In settings with multiple candidate models, \citet{ZhangH2023} proposed a different probabilistic model for ROBs,
and has been used for various problems in computational mechanics \cite{Zhang2024,Quek2025}.

This paper proposes a novel framework for probabilistic modeling of principal subspaces.
The stochastic ROMs constructed from such stochastic subspaces enable
improved characterization of model uncertainty in computational mechanics.
The main contributions of our work are as follows:
\begin{enumerate}
\item We introduce a new class of probabilistic models of subspaces,
    which is simple, has nice analytical properties and can be sampled efficiently.
    This is of general interest, and can be seen as a stochastic version of the proper orthogonal decomposition.
\item We use this stochastic subspace model---instead of a probabilistic model of the reduced-order basis---to construct stochastic ROMs,
    and reduce the number of hyperparameters to one.
\item The optimization of the hyperparameter is fully automated
  and done efficiently by exploiting one-dimensional optimization algorithms.
\item We demonstrate in a series of numerical experiments that our method provides consistent
  and sharp uncertainty estimates, with low computational costs.
\end{enumerate}

This paper is organized as follows:
A brief introduction to stochastic reduced-order modeling is presented in \cref{sec:SROM}.
The proposed stochastic subspace model, along with the algorithm, is described in \cref{sec:PPCA}.
Critical information regarding the hyperparameter optimization is given in \cref{sec:training}.
Related works in the literature are reviewed in \cref{sec:related}.
The accuracy and efficiency of the proposed method are validated using numerical examples in \cref{sec:examples}.
Finally, \cref{sec:conclusion} concludes the paper with a brief summary and potential future work.

\section{Stochastic reduced-order modeling}\label{sec:SROM}
This section introduces the concept of stochastic reduced-order modeling. 
We begin in \cref{sec:HDM} with the high-dimensional model (also known as the full-order model),
followed in \cref{sec:ROM} by its reduction to a lower-order representation. 
Building on this, \cref{sec:sub-SROM} presents a stochastic version of the reduced-order model.
The framework is presented in a general form, making it broadly applicable to various engineering systems.

\subsection{High-dimensional model}\label{sec:HDM}

In general, we consider a parametric nonlinear system
given by a set of ordinary differential equations (ODEs):
\begin{equation}\label{eq:HDM}
  \dot{\mathbf{x}} = \mathbf{f}(\mathbf{x}, t; \boldsymbol{\mu}),
\end{equation}
with $\mathbf{x} \in \mathbb{R}^n$ and $t \in [0, \infty)$,
and is subject to initial conditions $\mathbf{x}(0) = \mathbf{x}_0$
and linear constraints $\mathbf{B}^\intercal \mathbf{x} = 0$,
where $\mathbf{B} \in \mathbb{R}^{n \times n_{CD}}$.
Such constraints can capture, for example, boundary conditions.
If nonhomogeneous linear constraints are present, they can be transformed into homogeneous constraints by shifting the system states, without loss of generality.
These equations often come from a spatial discretization
of a set of partial differential equations
that govern a given physical system, with dimension $n \gg 1$.
We call this the \textit{high-dimensional model} (HDM).
If the system is time-independent, \cref{eq:HDM} reduces to a set of algebraic equations:
$\mathbf{f}(\mathbf{x}; \boldsymbol{\mu}) = 0$. 
The HDM does not fully represent the physical system due to approximations introduced at various stages of modeling and simulation. 
This difference is the target for model error analysis.

\subsection{Reduced-order model}\label{sec:ROM}

Reduced-order models (ROMs) can be constructed from the HDM to provide faster solutions
while maintaining accuracy.
A general approach to ROM construction is the Galerkin projection,
which projects the HDM orthogonally onto a $k$-dim subspace $\mathscr{V}$
of the state space $\mathbb{R}^n$.
This reduced subspace $\mathscr{V}$ can be obtained, for example,
using proper orthogonal decomposition (POD) \cite{Sirovich1987} applied to a snapshot matrix of HDM solutions, or from the eigendecomposition of the system operators \cite{Chopra1996}.
For POD, the reduced order $k$ is typically chosen by specifying a threshold for the fraction of total energy retained by the POD modes.
Let $\mathbf{V} \in \text{St}(n, k)$ be an orthonormal basis of the subspace $\mathscr{V}$,
then the ROM can be written as:
\begin{equation}\label{eq:ROM}
  \mathbf{x} = \mathbf{V} \mathbf{q},
  \quad
  \dot{\mathbf{q}} = \mathbf{V}^\intercal \, \mathbf{f}(\mathbf{V} \mathbf{q}, t; \boldsymbol{\mu}),
\end{equation}
with $\mathbf{q} \in \mathbb{R}^k$
and initial conditions $\mathbf{q}(0) = \mathbf{V}^\intercal \mathbf{x}_0$.
To satisfy the linear constraints, we must have $\mathbf{B}^\intercal \mathbf{V}= 0$. 
Because POD bases are typically constructed from snapshot data, the reduced-order model will inherently satisfy any linear constraints (homogeneous or non-homogeneous) present in the snapshots. 
For time-independent systems, \cref{eq:ROM} reduces to a set of algebraic equations:
$\mathbf{V}^\intercal \, \mathbf{f}(\mathbf{V} \mathbf{q}; \boldsymbol{\mu})$ = $0$. 
The ROM deviates from the HDM due to model truncation,
which is also a source of model error that requires analysis.

\subsection{Stochastic reduced-order model}\label{sec:sub-SROM}

The stochastic reduced-order model (SROM) builds upon the structure of ROM with the difference that
the deterministic ROB $\mathbf{V}$ is replaced by its stochastic counterpart $\mathbf{W}$.
The change from a deterministic basis to a stochastic one
introduces randomness into the model and allows it to capture
model errors between the ROM, the HDM, and the physical system.
The SROM can be written as:
\begin{equation}\label{eq:SROM}
  \mathbf{x} = \mathbf{W} \mathbf{q},
  \quad \dot{\mathbf{q}} = \mathbf{W}^\intercal \, \mathbf{f}(\mathbf{W} \mathbf{q}, t; \boldsymbol{\mu}),
  \quad \mathbf{W} \sim \mu_{\mathbf{V}}
\end{equation}
where stochastic basis $\mathbf{W}$ follows a probability distribution $\mu_{\mathbf{V}}$.
However, it is more natural to impose a probability distribution $\mu_{\mathscr{V}}$ on the subspace rather than the basis,
  because the Galerkin projection is uniquely determined by the subspace and is invariant under changes of basis.
In this work, we use the stochastic subspace model introduced in \cref{sec:PPCA}
to sample stochastic basis $\mathbf{W}$ and construct SROM, which is then used to characterize model error.

\section{Stochastic subspace model} \label{sec:PPCA}

This section presents the formulation of the stochastic subspace model using probabilistic principal component analysis.
The derivation of a deterministic principal subspace via proper orthogonal decomposition is outlined in \cref{sec:subspace-PCA}.
The procedure of constructing random matrices is described in \cref{sec:random-matrices-ppca}.
The stochastic subspace model and the sampling algorithm are discussed in \cref{sec:stochastic-subspace-ppca}.

\subsection{Subspace from the PCA: proper orthogonal decomposition} \label{sec:subspace-PCA}

Principal component analysis (PCA) is a common technique for dimension reduction.
Closely related to PCA, a common way to find the ROB $\mathbf{V}$
and the associated subspace $\mathscr{V}$ is the \textit{proper orthogonal decomposition} (POD).
Following the notations in \cref{sec:SROM},
let $\mathbf{X} = [\mathbf{x}_1 \, \cdots \, \mathbf{x}_m] \in \mathbb{R}^{n \times m}$
be a sample of the state, with $\mathbf{x}_i = \mathbf{x}(t_i; \boldsymbol{\mu}_i)$,
sample mean $\overline{\mathbf{x}} = \frac{1}{m} \sum_{i=1}^m \mathbf{x}_i$,
centered sample $\mathbf{X}_0 = \mathbf{X} - \overline{\mathbf{x}} \mathbf{1}_m^\intercal$,
and sample covariance $\mathbf{S} = \frac{1}{m} \mathbf{X}_0 \mathbf{X}_0^\intercal$.
Let $\mathbf{X}_0 = \mathbf{V}_r \diag(\boldsymbol{\sigma}_r) \mathbf{W}_r^\intercal$
be a compact singular value decomposition (SVD), where
$\boldsymbol{\sigma}_r \in \mathbb{R}^r_{>0 \downarrow}$ is in non-increasing order.
The sample covariance matrix $\mathbf{S}$ can be written as
$\mathbf{S} = \mathbf{V}_r \diag(\boldsymbol{\sigma}_r / \sqrt{m})^2 \mathbf{V}_r^\intercal$.
POD takes the \textit{principal basis} $\mathbf{V}_k$---the leading $k$ eigenvectors of $\mathbf{S}$---as the deterministic ROB.
The corresponding subspace is the \textit{principal subspace} $\mathscr{V}_k := \text{range}(\mathbf{V}_k)$.
Since all state samples satisfy the linear constraints,
$\mathbf{V}_k$ satisfies the constraints automatically.
POD has several other desirable properties: it can extract coherent structures, provide error estimates,
and is computationally efficient and straightforward to implement \cite{Sirovich1987}.

\subsection{Random matrices from the PPCA}\label{sec:random-matrices-ppca}

Probabilistic PCA (PPCA) \citep{Tipping1999}
reformulates PCA as the maximum likelihood solution
of a probabilistic latent variable model.
Here, we briefly overview PPCA and introduce some related random matrix models.

Assume that the data-generating process can be modeled as:
$\mathbf{x} = \boldsymbol{\mu} + \mathbf{U} \mathbf{z} + \boldsymbol{\epsilon}$,
where $\boldsymbol{\mu} \in \mathbb{R}^n$, $\mathbf{U} \in \mathbb{R}^{n \times k}$,
$\mathbf{z} \sim N_k(0, \mathbf{I}_k)$, $\boldsymbol{\epsilon} \sim N_n(0, \varepsilon^2 \mathbf{I}_n)$,
$n \in \mathbb{Z}_{>0}$ is the ambient dimension, and $k \in \{1, \cdots, n\}$ is the latent dimension.
Then the observed data follows a Gaussian distribution $\mathbf{x} \sim N_n(\boldsymbol{\mu}, \mathbf{C})$,
where covariance matrix $\mathbf{C} = \mathbf{U} \mathbf{U}^\intercal + \varepsilon^2 \mathbf{I}_n$.
The corresponding log-likelihood for $m$ samples 
$\{\mathbf{x}_i\}_{i=1}^m$ with sample covariance 
$\mathbf{S} = \frac{1}{m} \sum_{i=1}^m (\mathbf{x}_i - \boldsymbol{\mu})(\mathbf{x}_i - \boldsymbol{\mu})^\top$ 
is:
$L = -\frac{m}{2} \left[ n\ln(2\pi) + \ln|\mathbf{C}| + \mathrm{tr}(\mathbf{C}^{-1} \mathbf{S}) \right]$.

Assume that we have a data sample
$\mathbf{X} = [\mathbf{x}_1 \, \cdots \, \mathbf{x}_m] \in \mathbb{R}^{n \times m}$ of size $m$.
Let $\mathbf{S} = \mathbf{V} \diag(\boldsymbol{\lambda}) \mathbf{V}^\intercal$
be an eigenvalue decomposition (EVD) of the sample covariance matrix,
where $\mathbf{V} \in O(n)$ is an orthogonal matrix and
$\boldsymbol{\lambda} \in \mathbb{R}^n_{\ge 0 \downarrow}$ is in non-increasing order.
The maximum likelihood estimator of the model parameters $(\boldsymbol{\mu}, \mathbf{U}, \varepsilon^2)$ are:
$\widetilde{\boldsymbol{\mu}} = \overline{\mathbf{x}}$,
$\widetilde{\mathbf{U}} = \mathbf{V}_k [\diag(\lambda_i - \tilde{\varepsilon}^2)_{i=1}^k]^{1/2} \mathbf{Q}$,
and $\tilde{\varepsilon}^2 = \frac{1}{n-k} \sum_{i=k+1}^n \lambda_i$,
where $\mathbf{V}_k = [\mathbf{v}_1 \, \cdots \, \mathbf{v}_k]$ consists of the first $k$ columns of $\mathbf{V}$,
and $\mathbf{Q} \in O(k)$ is any order-$k$ orthogonal matrix.
The parameters $\widetilde{\mathbf{U}}$, $\tilde{\varepsilon}^2$  are obtained by iterative maximization of $L$ via the EM algorithm described in \cite{Tipping1999}.
The estimated covariance matrix can thus be written as
$\widetilde{\mathbf{C}} = \mathbf{V}_k \diag(\lambda_i - \tilde{\varepsilon}^2)_{i=1}^k \mathbf{V}_k^\intercal + \tilde{\varepsilon}^2 \mathbf{I}_n$.
We call the columns of $\mathbf{V}_k$ \textit{principal components} of the data sample
and $\boldsymbol{\lambda}_k = (\lambda_i)_{i=1}^k$ the corresponding \textit{principal variances}.

The PPCA assumes that the latent dimension $k$ is known,
and estimates the unknown data distribution $\mathbf{x} \sim N_n(\boldsymbol{\mu}, \mathbf{C})$
by maximum likelihood as $\widetilde{\mathbf{x}} \sim N_n(\widetilde{\boldsymbol{\mu}}, \widetilde{\mathbf{C}})$.
Without assuming $k$ as known, the data distribution can simply be modeled as
$\widetilde{\mathbf{x}} \sim N_n(\overline{\mathbf{x}}, \mathbf{S})$,
the multivariate Gaussian distribution parameterized by the sample mean and the sample covariance.
From this empirical data distribution, we can derive distributions of some related matrices,
all in analytical form.

First, the data sample matrix has a matrix-variate Gaussian distribution
$\widetilde{\mathbf{X}} = [\widetilde{\mathbf{x}}_1 \, \cdots \, \widetilde{\mathbf{x}}_m]
\sim N_{n,m}(\overline{\mathbf{x}}; \mathbf{S}, \mathbf{I}_m)$,
where the columns are independent and identically distributed as the empirical data distribution
$\widetilde{\mathbf{x}}_i \overset{\text{iid}}{\sim} N_n(\overline{\mathbf{x}}, \mathbf{S})$,
$i = 1, \cdots, m$.
The orientation of the data sample matrix of size $k$ has a
matrix angular central Gaussian (MACG) distribution~\cite{Chikuse2003}:
\begin{equation}\label{eq:MACG-St}
  \widetilde{\mathbf{V}}_k := \pi(\widetilde{\mathbf{X}}_{0,k}) \sim \text{MACG}_{n,k}(\mathbf{S}),
\end{equation}
where $\widetilde{\mathbf{X}}_{0,k} \sim N_{n,k}(0; \mathbf{S}, \mathbf{I}_k)$
is a centered $k$-sample matrix
and $\pi(\mathbf{M}) := \mathbf{M} (\mathbf{M}^\intercal \mathbf{M})^{-1/2}$
is \textit{orthonormalization by polar decomposition}.
As a mapping, $\pi: \mathbb{R}^{n \times k}_* \mapsto \text{St}(n, k)$
is uniquely defined for all full column-rank $n$-by-$k$ matrices, $k \le n$,
and takes values in the \textit{Stiefel manifold}
$\text{St}(n, k) := \{\mathbf{V} \in \mathbb{R}^{n \times k} :
\mathbf{V}^\intercal \mathbf{V} = \mathbf{I}_k\}$.
The range (i.e., column space) of the data sample matrix of size $k$
also has an MACG distribution:
\begin{equation}
  \widetilde{\mathscr{V}}_k := \text{range}(\widetilde{\mathbf{V}}_k)
  \sim \text{MACG}_{n,k}(\mathbf{S}),
\end{equation}
which is supported on the \textit{Grassmann manifold} $\text{Gr}(n, k)$
of $k$-dim subspaces of the Euclidean space $\mathbb{R}^n$.
Its probability density function (PDF) has a unique mode---and therefore
a unique global maximal point---at the principal subspace $\mathscr{V}_k$, see \ref{apd:mode-proof}.
This distribution is useful, for example, for modeling subspaces for parametric ROM \cite{ZhangRD2022gps}.

\subsection{Stochastic subspace via the PPCA}\label{sec:stochastic-subspace-ppca}

Here we propose a new class of stochastic subspace models $\text{MACG}_{n,k,\beta}(\mathbf{S})$
with $\beta \in [k, \infty)$.
The classical $\text{MACG}_{n,k}(\mathbf{S})$ distribution is a special case with $\beta = k$.

Define the \textit{principal subspace map} $\pi_k(\mathbf{X}) := \text{range}(\mathbf{U}_k)$,
where $\mathbf{U}_k$ consists of singular vectors associated with
the $k$ largest singular values of $\mathbf{X}$.
As a mapping, $\pi_k: \mathbb{R}^{n \times m}_{k>} \mapsto \text{Gr}(n, k)$ is uniquely defined
for all $n$-by-$m$ matrices whose $k$-th largest singular value is larger than the $(k+1)$-th,
with $k \le \min(n, m)$.
Using this map, the POD subspace can be written as:
  $\mathscr{V}_k = \pi_k(\mathbf{X}_0)$.
We define a stochastic subspace model:
\begin{equation}\label{eq:MACG-nkp}
  \mathscr{W} := \pi_k(\widetilde{\mathbf{X}}_{0,\beta})
  \overset{\text{def}}{\sim} \text{MACG}_{n,k,\beta}(\mathbf{S}),
\end{equation}
where $\widetilde{\mathbf{X}}_{0,\beta} \sim N_{n,\beta}(0; \mathbf{S}, \mathbf{I}_{\beta})$
is a centered $\beta$-sample matrix and $\beta \in \{k, k+1, \cdots\}$ is the resample size.
Considering its close connection with the POD, we may call this model a \textit{stochastic POD}.
Similar to $\text{MACG}_{n,k}$,
the unique mode and global maximal of $\text{MACG}_{n,k,\beta}$
is the principal subspace $\mathscr{V}_k$.
The hyperparameter $\beta$ controls the concentration of the distribution:
larger value means less variation around $\mathscr{V}_k$.
As with POD, the stochastic basis $\mathbf{W}$ sampled from the stochastic subspace model
defined in \cref{eq:MACG-nkp} satisfies the linear constraints
$\mathbf{B}^\intercal \mathbf{W} =0$ automatically.

Sampling from this model
can be done very efficiently for $n \gg 1$
if $k$ and the rank of $\mathbf{S}$ are small,
which is often the case in practical applications.
Let $\mathbf{S} = \mathbf{V}_r \diag(\boldsymbol{\lambda}_r) \mathbf{V}_r^\intercal$
be a compact EVD where $r = \text{rank}(\mathbf{S})$.
We can show that:
\begin{equation}\label{eq:MACG-nkp-low-rank}
  \mathscr{W} = \mathbf{V}_r\, \pi_k\left(\diag(\boldsymbol{\lambda}_r)^{1/2}\,
    \mathbf{Z}_{r \times \beta}\right),
\end{equation}
where $\mathbf{Z}_{r \times \beta}$ is an $r$-by-$\beta$ standard Gaussian matrix.
\Cref{alg:ss-ppca} gives a procedure for sampling $\text{MACG}_{r,k,\beta}(\mathbf{S})$
with $\mathbf{S} = \diag(\mathbf{s})^2$,
which in effect implements the map $\pi_k(\diag(\mathbf{s})\,\mathbf{Z}_{r \times \beta})$.
For a general covariance matrix $\mathbf{S}$, an orthonormal basis of a random subspace
sampled from $\text{MACG}_{n,k,\beta}(\mathbf{S})$ can be obtained by
left multiplying the output of \Cref{alg:ss-ppca} with $\mathbf{V}_r$.
Using \cref{eq:MACG-nkp-low-rank} instead of \cref{eq:MACG-nkp}
reduces the computational cost for truncated SVD from $O(n k \beta)$ to $O(r k \beta)$.

\alglanguage{pseudocode}
\begin{algorithm}[t]
  \caption{\texttt{SS-PPCA}: Stochastic subspace via probabilistic principal component analysis.}
  \label{alg:ss-ppca}
  \begin{algorithmic}[1] %
    \Input scale vector $\mathbf{s} \in \mathbb{R}^r_{>0}$;
    subspace dimension $k \in \{1, \cdots, r\}$;
    resample size $\beta \in \{k, k+1, \cdots\}$.
    \State Generate $\mathbf{Z}_{r \times \beta} \in \mathbb{R}^{r \times \beta}$ with entries
    $z_{ij} \overset{\text{iid}}{\sim} N(0, 1)$
    \State $\mathbf{M} \gets \diag(\mathbf{s})\, \mathbf{Z}_{r \times \beta}$
    \State Truncated SVD: $[\mathbf{U}_k, \mathbf{d}_k, \mathbf{V}_k] \gets \text{svd}(\mathbf{M}, k)$
    \Output $\mathbf{W} = \mathbf{U}_k$,
    an orthonormal basis of a random subspace sampled from $\text{MACG}_{r,k,\beta}(\mathbf{S})$
    with $\mathbf{S} = \diag(\mathbf{s})^2$.
  \end{algorithmic}
\end{algorithm}

We can generalize the concentration hyperparameter $\beta$ to real numbers in $[k, \infty)$,
which leads to a continuous family of distributions $\text{MACG}_{n,k,\beta}(\mathbf{S})$
and can be useful when the optimal value of $\beta$ is not much greater than one.
In this case, we sample $\mathbf{Z} \in \mathbb{R}^{r \times \lceil \beta \rceil}$,
where $\lceil \beta \rceil$ denotes the smallest integer greater than or equal to $\beta$.
To account for the effect of real-valued $\beta$,
the weight of the final column of $\mathbf{Z}$ is set to $\beta-\lfloor \beta\rfloor$,
where $\lfloor \beta \rfloor$ is the largest integer less than or equal to $\beta$.
The rest of the steps are the same as \cref{alg:ss-ppca}.

\subsection{Construction of SROM using SS-PPCA}
The main steps of the proposed approach to stochastic reduced-order modeling are listed below:
\begin{enumerate}
\item \textbf{Solve HDM and collect snapshots:} Run the HDM and store state vectors at selected times and/or parameters to form the snapshot matrix $\mathbf{X} = [\mathbf{x}_1, \ldots, \mathbf{x}_m]$.
\item \textbf{Extract deterministic ROB:} Center the snapshots $\mathbf{X}_0 = \mathbf{X} - \overline{\mathbf{x}} \mathbf{1}_m^\intercal$, where $\overline{\mathbf{x}} =\frac{1}{m} \sum_{i=1}^m \mathbf{x}_i$.
Perform a compact SVD: $\mathbf{X}_0 = \mathbf{V}_r \diag(\boldsymbol{\sigma}_r) \mathbf{W}_r^\intercal$. 
Select $k$ as the smallest integer $j$ satisfying
$\sum_{i=1}^{j} \sigma_i^2 \ge \tau \sum_{i=1}^{r} \sigma_i^2$,
where $\tau \in (0,1)$ is a prescribed variance threshold.  
Set the deterministic ROB to the first $k$ singular vectors $\mathbf{V}_k$.
\item \textbf{Solve ROM:} Substitute $\mathbf{V}_k$ into \cref{eq:ROM} to obtain the deterministic ROM solution.
\item \textbf{Tune stochastic subspace model:} Using the SS-PPCA \cref{alg:ss-ppca},
generate stochastic subspaces $\mathscr{W} \sim \text{MACG}_{n,k,\beta}(\mathbf{S})$
represented by bases $\mathbf{W} = \mathbf{V}_r \mathbf{U}_k$,
and tune the concentration parameter $\beta \in [k, \infty)$ according to the criterion in \cref{eq:objective}.
\item \textbf{Build SROM ensemble:} With the tuned $\beta$, sample $\mathscr{W}$ to produce stochastic bases $\mathbf{W}$ to replace $\mathbf{V}_k$ in \cref{eq:ROM}. The resulting ensemble of SROM realizations quantifies the model and truncation error through the variability of their outputs.
\end{enumerate}

These steps are applicable to both linear and nonlinear systems. 
For linear systems, however, the cost of SROM construction can be reduced substantially.
Naively, each SROM realization requires constructing a stochastic basis
$\mathbf{W} = \mathbf{V}_r \mathbf{U}_k$ and projecting the HDM onto this basis.
This cost accumulates with the number of SROM samples. 
A more efficient strategy is to perform a two-stage reduction:
first project the operators once onto the rank-$r$ subspace spanned by $\mathbf{V}_r$
(e.g., $\mathbf{A}_r = \mathbf{V}_r^\intercal \mathbf{A} \mathbf{V}_r$ for
any system matrix $\mathbf{A}$);
and then for each stochastic basis $\mathbf{U}_k$ sampled in this $r$-dimensional subspace,
compute $\mathbf{A}_{\mathbf{W}} = \mathbf{U}_k^\intercal \mathbf{A}_r \mathbf{U}_k$
so that the overall basis is $\mathbf{W} = \mathbf{V}_r \mathbf{U}_k$. 
This two-stage reduction avoids repeated high-dimensional operator projections
and significantly lowers the SROM construction cost.

\section{Hyperparameter training} \label{sec:training}

To find the optimal hyperparameter $\beta \in [k, \infty)$, we minimize the following objective function:
\begin{equation} \label{eq:objective}
  f(\beta) := \mathbb{E}[|d_o(\mathbf{u}_L) - d_o(\mathbf{u}_E)|^2 \, | \, \beta],
\end{equation}
where $\mathbf{u}_E$ is the experimental or ground-truth observation of the output,
$\mathbf{u}_L$ is the low-fidelity prediction of the SROM,
and $d_o(\mathbf{u}) := \|\mathbf{u} - \mathbf{u}_L^o\|_{L^2}$ is the $L^2$ distance
to the low-fidelity prediction $\mathbf{u}_L^o$ of a reference model.
Given the value of the hyperparameter $\beta$, the stochastic subspace model determines the stochastic ROM,
which produces stochastic predictions $\mathbf{u}_L$ that is then summarized into the random variable $d_o(\mathbf{u}_L)$.
Given the experimental or ground-truth observation $\mathbf{u}_E$,
the objective function measures the mean squared error between $d_o(\mathbf{u}_L)$ and $d_o(\mathbf{u}_E)$,
which is a statistical measure of how closely the SROM resembles the ground truth.
Overall, this optimization problem aims to improve the consistency of the SROM
in characterizing the error of the reference model.

We optimize the objective function $f(\beta)$ using a one-dimensional optimization algorithm,
such as golden section search and successive parabolic interpolation.
Implementations of such algorithms are readily available in many programming languages,
e.g., \texttt{scipy.optimize.minimize_scalar()} in Python,
\texttt{fminbnd()} in Matlab, and \texttt{optimize()} in R.
The objective function, however, is not accessible due to the expectation over SROM.
To reduce computation, we approximate $f(\beta)$ at integer $\beta$ values by Monte Carlo sampling,
and linearly interpolate $f(\beta)$ between consecutive integer points.
Unless specified otherwise, 1,000 random samples of the SROM are used for each $\beta$ value.
In practice, when the optimization algorithm queries $f(\beta)$ at a real-valued $\beta$,
its values at the two closest integer points will be computed using Monte Carlo sampling and stored.
If future queries need the value of $f(\beta)$ at a previously evaluated integer point,
the stored value will be used to avoid re-evaluation.

The optimization scheme employed in this study is very easy to implement.
A possible limitation is its computational efficiency
due to the Monte Carlo approximation of the expectation operator.
Although SROMs are computationally inexpensive to simulate compared with the HDM,
evaluating thousands of SROM samples can still be moderately costly.
Advanced optimization algorithms may be designed to address this issue,
but it is out of scope for the current work.

\section{Related works} \label{sec:related} %

This section reviews in detail the two main works that use SROM for model error characterization.
Firstly, the work by \citet{Soize2017srob} is discussed in \cref{sec:NPM}.
The work by \citet{ZhangH2023} for characterizing model-form uncertainty across multiple models is presented in \cref{sec:RSM}.
In \cref{sec:compare-model-structures}, we distinguish the three methods
based on their operational mechanisms and their dependence on available information.

\subsection{Non-parametric model (NPM)} \label{sec:NPM}

The seminal work of NPM \citep{Soize2017srob} introduced the idea of using SROM to analyze model-form uncertainty.
The stochastic basis in the NPM can be written as:
\begin{equation}
  \mathbf{W} = \pi\left(\overline{\mathbf{V}} +
    \mathcal{P}^{\text{St}}_{\overline{\mathbf{V}}}(\mathbf{P}_{\mathscr{B}^\perp} \mathbf{U})\right),
  \quad \mathbf{U} = s \mathbf{G} \mathbf{R},
  \quad \mathbf{G} \sim \mathcal{GP}(\gamma).
\end{equation}
The random matrix $\mathbf{G}$ requires knowledge of the discretization $\mathscr{D}$ of the spatial fields,
and is sampled from a Gaussian process with a hyperparameter $\gamma$ for correlation length-scale.
The column dependence matrix $\mathbf{R}$
is an order-$k$ upper triangular matrix with positive diagonal entries.
The scale $s \in [0, \infty)$ is a non-negative number.
Once $\mathbf{U}$ is constructed,
orthogonal projection $\mathbf{P}_{\mathscr{B}^\perp}$ enforces the linear constraints,
tangential projection $\mathcal{P}^{\text{St}}_{\overline{\mathbf{V}}}$
projects a matrix to the tangent space of the Stiefel manifold
at a reference ROB $\overline{\mathbf{V}}$,
and orthonormalization $\pi$ as in \cref{eq:MACG-St} gives an orthonormal basis.
The hyperparameters $(s, \gamma, \mathbf{R})$ have a dimension of $k(k+1)/2 + 2$,
and are trained by minimizing a weighted average of $J_{\text{mean}}$ and $J_{\text{std}}$:
the former aims at matching the SROM mean prediction to ground-truth observations $\mathbf{O}$;
the latter aims at matching the SROM standard deviation to a scaled difference
between $\mathbf{O}$ and $\mathbf{u}_L^o$.

\subsection{Riemannian stochastic model (RSM)} \label{sec:RSM}

For problems where there exist a number of physically plausible candidate models,
the RSM \citep{ZhangH2023} characterizes model-form uncertainty
using SROMs constructed with the following stochastic basis:
\begin{equation}
  \mathbf{W} = \exp^{\text{St}}_{\overline{\mathbf{V}}}%
  \Big\{c \sum_{i=1}^q p_i
  \log^{\text{St}}_{\overline{\mathbf{V}}}(\mathbf{V}^{(i)})%
  \Big\},
  \quad \mathbf{p} \sim \text{Dirichlet}(\boldsymbol{\alpha}),
\end{equation}
where $c \in [0, \infty)$ is a scale parameter,
$q$ is the number of models,
and $\mathbf{p} = (p_i)_{i=1}^q$ is a probability vector
following the Dirichlet distribution
with concentration parameters $\boldsymbol{\alpha} \in \mathbb{R}_{>0}^q$.
The Riemannian logarithm $\log^{\text{St}}_{\overline{\mathbf{V}}}$ maps a model-specific ROB $\mathbf{V}^{(i)}$
to the tangent space of the Stiefel manifold at a reference ROB $\overline{\mathbf{V}}$.
The Riemannian exponential $\exp^{\text{St}}_{\overline{\mathbf{V}}}$ maps a tangent vector
at $\overline{\mathbf{V}}$ back to the Stiefel manifold,
giving the orthonormal basis $\mathbf{W}$.
The scale parameter $c$ is usually set to one or omitted in later works \cite{Zhang2024,Quek2025}.
The hyperparameters $\boldsymbol{\alpha}$ are trained by minimizing
$\|\mathbb{E}[\log^{\text{St}}_{\overline{\mathbf{V}}} \mathbf{W}]\|_F^2$,
to match the center of mass of the distribution to $\overline{\mathbf{V}}$.
It is a quadratic program and can be solved efficiently.
While the RSM is designed for multi-model settings,
finding a reasonable way to apply it to a single-model setting
is an open problem that is to be investigated by the original authors.

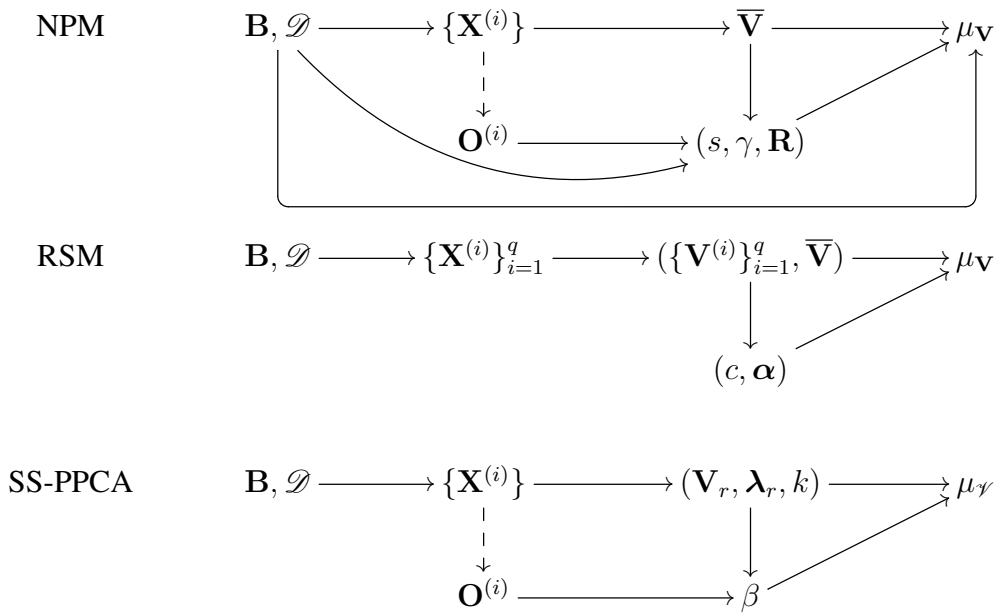
\begin{figure}[t]
  \centering
  \[
    \xymatrix@C+1pc{
      \text{NPM} & \mathbf{B}, \mathscr{D}
      \ar@/^-1.15cm/[drr]       %
      \ar[r]
      \ar `d/4pt[dr]-<0em,2em> `/4pt[rrr] [rrr] %
      & \{\mathbf{X}^{(i)}\} \ar[r] \ar@{-->}[d] & \overline{\mathbf{V}} \ar[r] \ar[d] & \mu_\mathbf{V} \\
      &  & {\mathbf{O}^{(i)}} \ar[r] & (s, \gamma, \mathbf{R}) \ar[ur] & \\
      \text{RSM} &  \mathbf{B}, \mathscr{D} \ar[r] & \{\mathbf{X}^{(i)}\}^q_{i = 1} \ar[r] & (\{\mathbf{V}^{(i)}\}^q_{i=1} , \overline{\mathbf{V}}) \ar[r] \ar[d] & \mu_\mathbf{V} \\
      & & & (c, \boldsymbol{\alpha}) \ar[ur] & \\
      \text{SS-PPCA} & \mathbf{B}, \mathscr{D} \ar[r] & \{\mathbf{X}^{(i)}\} \ar[r] \ar@{-->}[d] & (\mathbf{V}_r, \boldsymbol{\lambda}_r, k) \ar[r] \ar[d] & \mu_{\mathscr{V}} \\
      & & {\mathbf{O}^{(i)}} \ar[r] & \beta \ar[ur] & \\
    }
  \]
  \caption{Dependence diagrams for (from top to bottom)
    the NPM and the RSM models of stochastic basis
    and the SS-PPCA model of stochastic subspace.
    The dashed lines indicate that the connection only exists
    when characterizing the ROM-to-HDM error.}
  \label{fig:comparision}
\end{figure}

\subsection{Comparison of model structures} \label{sec:compare-model-structures}

To deepen understanding of the three methods, we outline their structures in \Cref{fig:comparision}
and highlight a few key differences below.
Both NPM and RSM construct probabilistic models of reduced-order bases, whereas SS-PPCA models subspaces directly.
Another distinction lies in the reliance on discretization information:
NPM requires discretization and boundary condition data for hyperparameter optimization and SROM sampling,
while SS-PPCA and RSM do not.
Finally, in default settings, RSM optimizes its hyperparameters using only the reduced bases,
while SS-PPCA and NPM incorporate observed quantities of interest into their objective functions.
Other objective functions may be developed to incorporate observed quantities for the calibration of RSM hyperparameters.

\section{Numerical experiments} \label{sec:examples}

In this section, we evaluate the proposed method, SS-PPCA, in characterizing model error using three numerical examples.
In \cref{sec:param-static}, we examine a parametric nonlinear static problem and characterize model error at unseen parameters.
In \cref{sec:static}, we demonstrate how our approach can be used to
characterize the HDM error, assuming that experimental data is available.
Finally, in \cref{sec:dynamic}, we discuss a linear dynamics problem involving a space structure.
All data and code for these examples are available at \href{https://github.com/UQUH/SS_PPCA}{https://github.com/UQUH/SS_PPCA}.

The performance of the SROMs is assessed using two general metrics for stochastic models: consistency and sharpness.
By \textit{consistent} UQ,
we mean that the statistics on prediction uncertainty---such as standard deviation
and predictive intervals (PI)---derived from the model match the prediction error on average.
To measure consistency, we calculate the coverage, defined as the percentage of experimental observations captured by the PIs. For the predictions to be consistent, the coverage percentage should approximately match the nominal confidence level used to construct the PIs.
Without consistency, the model either under- or over-estimates its prediction error,
which risks being over-confident or unnecessarily conservative. %
By \textit{sharp} UQ, we mean the model's PIs are narrower than alternative methods, which can be assessed quantitatively by computing the PI widths or qualitatively through visual inspection.
This is essentially the probabilistic version of model accuracy
and is desirable as we want to minimize error.

\subsection{Parametric nonlinear static problem} \label{sec:param-static}
We consider a one-dimensional parametric nonlinear static problem
with $n = \num{1000}$ degrees of freedom (DoFs), governed by the system of equations:
\begin{equation}
  \mathbf{K} \mathbf{x}(\boldsymbol{\mu}) + \alpha \mathbf{x}^3(\boldsymbol{\mu}) = \mathbf{f}(\boldsymbol{\mu}),
\end{equation}
where $\mathbf{x}(\boldsymbol{\mu}) \in \mathbb{R}^n$ is the displacement vector,
and $\boldsymbol{\mu} \in [0,1]^5$ is a parameter vector which controls the loading conditions.
The parameter $\alpha > 0$ controls the strength of the cubic nonlinearity. In this example, we set $\alpha = 10^4$.
The system is subject to homogeneous Dirichlet boundary conditions
at the first and the last node, i.e., $x_1 (\boldsymbol{\mu}) = x_n (\boldsymbol{\mu})= 0$.
These constraints can be compactly represented as
$\mathbf{B}^\intercal \mathbf{x}(\boldsymbol{\mu}) = 0$
with $\mathbf{B} = [\mathbf{e}_1 \, \mathbf{e}_n]$,
where $\mathbf{e}_i$ is the $i$-th standard basis vector of $\mathbb{R}^n$.
The stiffness matrix $\mathbf{K} \in \mathbb{R}^{n \times n}$ is constructed as
$\mathbf{K}$ = $\mathbf{\Phi} \mathbf{\Lambda} \mathbf{\Phi}^{\intercal}$, with
$\mathbf{\Lambda} = \diag(4 \pi^2 j^2)_{j=1}^{n-2}$ and
$\mathbf{\Phi} = [0 \; \mathbf{S} \; 0]^\intercal$.
The matrix $\mathbf{S}$ = $\sqrt{\frac{2}{n-1}} \left[\sin \left(\frac{jk \pi}{n-1}\right)\right]_{k=1, \dots, n-2}^{j=1, \dots, n-2}$
is the order-$(n-2)$ type-I discrete sine transform (DST-I) matrix, scaled to be orthogonal.
We can see that $\boldsymbol{\Lambda} \in \mathbb{R}^{(n-2) \times (n-2)}$ and $\boldsymbol{\Phi} \in \mathbb{R}^{n \times (n-2)}$.
Therefore $\mathbf{K}$ has eigenpairs $(\lambda_j,\boldsymbol{\phi}_j)$
with $\lambda_j = 4\pi^2j^2$ and $\boldsymbol{\phi}_j$ the $j$-th column of $\boldsymbol{\Phi}$.
The force is given by
$\mathbf{f}(\boldsymbol{\mu}) = \|\textbf{g}(\boldsymbol{\mu})\|_{\infty}^{-1} \textbf{g}(\boldsymbol{\mu})$
with $\textbf{g}(\boldsymbol{\mu}) = \sum_{i=2}^{6} \mu_i \boldsymbol{\phi}_i$ and $\mu_i \in [0,1].$

\begin{figure*}[!t]
  \centering
    \centering
    \includegraphics[width=\textwidth]{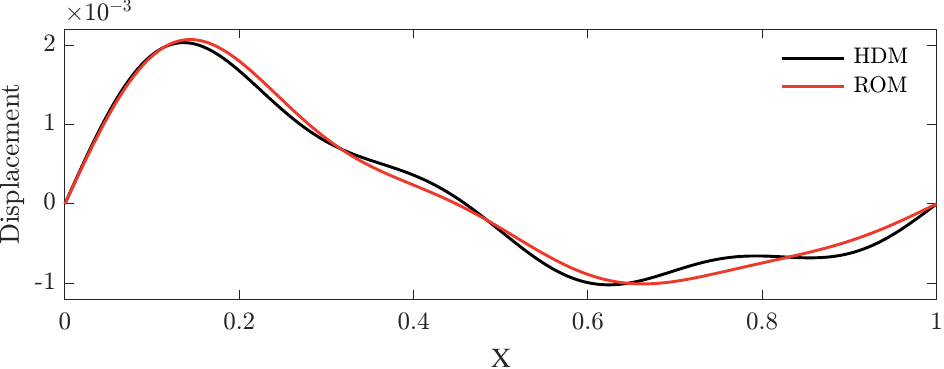}
    \caption{HDM vs ROM displacement at the test parameter}
    \label{fig:nonlinear-static-HDM-ROM}
\end{figure*}

The nonlinear HDM can be computationally expensive,
especially if it needs to be evaluated multiple times.
This justifies the use of ROM.
We construct a ROM by evaluating the HDM at 100 parameter samples,
$\boldsymbol{\mu}_j \in [0,1]^5$, drawn via Latin Hypercube Sampling (LHS).
For each sample, the solution $\mathbf{x}_j = \mathbf{x}(\boldsymbol{\mu}_j)$, $j = 1, \cdots, 100$ is computed by solving the nonlinear system using the Newton-Raphson method.
These solutions are concatenated to form a snapshot matrix $\mathbf{X}$,
which is used to construct a deterministic ROB $\mathbf{V} = \pi_k(\mathbf{X})$
with dimension $k=4$ via POD.
The governing equation of the ROM is given by:
\begin{equation}
    \mathbf{x}_{R}(\boldsymbol{\mu}) = \mathbf{V}\mathbf{q}(\boldsymbol{\mu}),
    \quad \mathbf{V}^{\intercal} \mathbf{K} \mathbf{V} \, \mathbf{q}(\boldsymbol{\mu}) + \alpha \, \mathbf{V}^\top \left( \left( \mathbf{V} \, \mathbf{q}(\boldsymbol{\mu}) \right)^{3} \right) = \mathbf{V}^\top \mathbf{f}(\boldsymbol{\mu})
\end{equation}

To evaluate the accuracy of this ROM, we define a test parameter as
$\boldsymbol{\mu}_{\text{test}} = [\frac{1}{2},\frac{1}{2},\frac{1}{2},\frac{1}{2},1]$,
which lies within the parameter domain but was not included in the training set.
The goal is to characterize the ROM error at this previously unseen parameter point,
providing insight into the model’s generalization capabilities.
\Cref{fig:nonlinear-static-HDM-ROM} shows the error between HDM and ROM displacement at the test parameter.

To characterize the error induced by ROM,
we construct SROMs by replacing the deterministic ROB $\mathbf{V}$ with SROBs $\mathbf{W}$.
The SROBs $\mathbf{W}$ are sampled using the stochastic subspace model
described in \cref{sec:stochastic-subspace-ppca} and \cref{alg:ss-ppca}.
To achieve consistent UQ, the hyperparameter $\beta$ is optimized
by minimizing the objective function given in \cref{eq:objective}.
For this example, $\mathbf{u}_E$ represents the HDM displacement at training parameters,
and $\mathbf{u}_L^o$ represents the ROM displacement at the same parameters.
The optimal value of the hyperparameter $\beta$ is 21, obtained using the strategy described in \cref{sec:training}.
The hyperparameter training time for this example is 1,480 seconds.
The long training time is primarily due to the nonlinear nature of the problem. Although the ROM is generally significantly faster than the HDM, the cubic nonlinearity requires projecting the ROM solution back to the full-order space, which is computationally intensive.
Employing efficient strategies for nonlinear model reduction \cite{Chaturantabut2010,Carlberg2010,Farhat2015} could substantially reduce the computational cost of hyperparameter training.

\begin{figure}[!t]
  \centering
  \includegraphics[width=\linewidth]{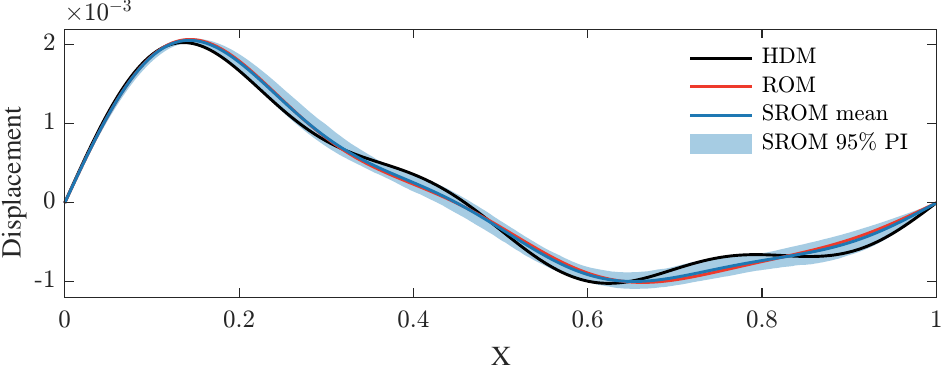}
  \caption{Nonlinear static problem: prediction by SS-PPCA.}
  \label{fig:nonlinear-PPCA-prediction}
\end{figure}

\Cref{fig:nonlinear-PPCA-prediction}
shows the ROM error characterization for the nonlinear parametric system at the test parameters.
The 95\% PI is estimated using 1,000 Monte Carlo samples of the SROM, both here and in the later examples.
It can be observed that the SROM mean displacement is very close to the ROM displacement, 
and the SROM 95\% PI is able to capture the ROM error consistently and sharply. 
The coverage for this example is 96.6\%, that is, the SROM 95\% PI has captured the ROM error consistently. 
It can also be observed that this coverage is achieved while maintaining tight bounds around the HDM displacement.

\subsection{Static problem: characterizing HDM error} \label{sec:static}

In the previous example,
we demonstrated the use of SROM to characterize the model truncation error induced by ROM in a nonlinear parametric system.
However, the utility of SROM is not limited to characterizing ROM error.
It can also be employed to characterize discrepancies between
HDM and experimental data, a crucial task in model validation.
Here, we characterize HDM error via the SROM approach in a linear static problem.

The synthetic experimental data is generated by solving the system:
\begin{equation}
  \mathbf{K}_{\text{E}} \mathbf{x}_{\text{E}} = \mathbf{f},
\end{equation}
with $n$ (model dimension) $= \num{1000}$,
$\mathbf{f} = \|\textbf{g}\|_{\infty}^{-1} \textbf{g}$
where $\textbf{g} = \sum_{i = 2}^{5} 0.5 \boldsymbol{\phi}^i + 1\boldsymbol{\phi}^6$
and homogeneous Dirichlet boundary conditions $\mathbf{B}^\intercal \mathbf{x}_{\text{E}} = 0$
where $\mathbf{B} = [\mathbf{e}_1 \, \mathbf{e}_n]$.
The stiffness matrix is defined as $\mathbf{K}_{\text{E}} = \mathbf{K} + \mathbf{K}_{\epsilon}$,
where $\mathbf{K}$ is same as the previous example
and $\mathbf{K}_{\epsilon}$ is
designed to induce a 15\% change in the Frobenius norm of $\mathbf{K}$.
In particular, 
$\mathbf{K}_{\epsilon}= \boldsymbol{\Phi}\boldsymbol{\Lambda}_{\epsilon} \boldsymbol{\Phi}^\intercal$ with $\boldsymbol{\Lambda}_{\epsilon} = \diag (c z_j \lambda_j)_{j=1}^{n-2}$, $z_j \overset{iid}{\sim} N(0,1)$, and $c$ is chosen such that $\| \boldsymbol{\Lambda}_{\epsilon}\|_{\text{F}} = 0.15  \|\boldsymbol{\Lambda}\|_{\text{F}}$. 
In addition to the model error, we simulate measurement error by adding 5\% noise
to the experimental displacement $\mathbf{x}_{\text{E}}$.
Hence, the observed experimental data follow:
$\mathbf{x}_{\text{obs}} = \mathbf{x}_{\text{E}} + \eta$, where $\eta \sim N(0,\sigma^2 \mathbf{I})$
with variance chosen to achieve the specified noise level.
Experimental data can be very expensive to acquire.
Therefore, we limit ourselves to experimental data from a sparse subset of spatial locations.
In this example, we pick 19 equidistant points within 0 and 1, excluding the boundaries.
These sparse measurements are used to characterize HDM error via the SROM approach.

\begin{figure}[!t]
  \centering
  \includegraphics[width=\textwidth]{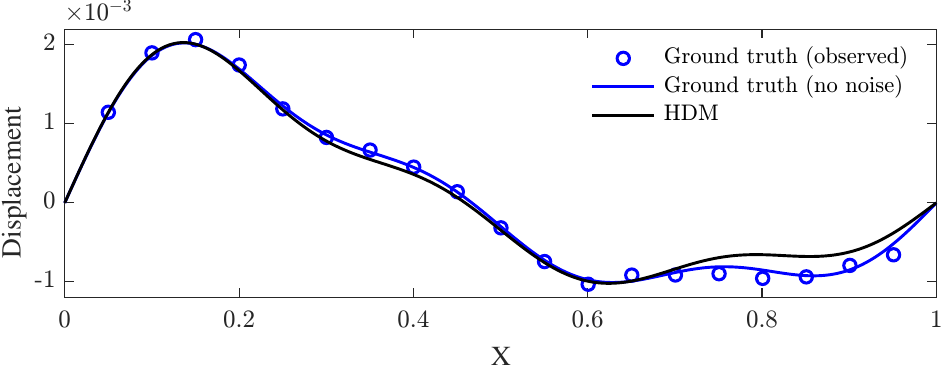}
  \caption{Ground truth vs HDM displacement.}
  \label{fig:ex3-HDM-EXP}
\end{figure}

The computational HDM is defined as:
\begin{equation}
  \mathbf{K} \mathbf{x}_{\text{H}} = \mathbf{f},
\end{equation}
with the same $\mathbf{K}$ and $\mathbf{f}$ as above.
However, the HDM produces biased predictions because the true system differs from this model, see \Cref{fig:ex3-HDM-EXP}.

To capture this error using SROM,
we collect snapshots with perturbed HDMs
so that the perturbed HDM responses can sufficiently cover the experimental data.
The stiffness matrix of the HDM is perturbed 100 times in a way similar to the construction of
$\mathbf{K}_\text{E}$, and the responses are collected in a snapshot matrix.
ROM of dimension $k = 4$ is then constructed using POD, which can be written as:
\begin{equation}
  \mathbf{x}_{\text{R}} = \mathbf{V}\mathbf{q},
  \quad \mathbf{V}^{\intercal} \mathbf{K} \mathbf{V} \mathbf{q} = \mathbf{V}^{\intercal}\mathbf{f}.
\end{equation}
We construct SROMs by substituting the deterministic ROB $\mathbf{V}$ with SROB $\mathbf{W}$.
The SROB $\mathbf{W}$ is sampled
using the stochastic subspace model described in \cref{sec:stochastic-subspace-ppca} and \cref{alg:ss-ppca}.
To ensure consistent UQ, the hyperparameter $\beta$ is optimized
by minimizing the objective function given in \cref{eq:objective}.
For this example, $\mathbf{u}_E$ represents the sparsely observable experimental displacement,
and $\mathbf{u}_L^o$ denotes the ROM displacement at the same locations.
Using the strategy given in \cref{sec:training},
the optimum value of the hyperparameter $\beta$ comes out to be 5.
However, since the optimal value of the hyperparameter is a small integer,
we may extend $\beta$ to real values to achieve better consistency.
In fact, the coverage of the 95\% PI using $\beta = 5$ is 82.2\%, notably short of the nominal 95\% value.
The main reason behind this low consistency is that the strategy discussed in \cref{sec:training}
evaluates $f(\beta)$ at integer $\beta$ and uses linear interpolation
to approximate $f(\beta)$ between two consecutive integer $\beta$,
which may not give the best results if the optimal hyperparameter is close to one.

To improve SROM consistency, the hyperparameter training for this example
follows a two-step procedure %
to accommodate a real-valued hyperparameter.
Step one follows the method described in \cref{sec:training}, using 1,000 Monte Carlo samples.
This initial search yields an optimal hyperparameter value of $\beta = 5$,
with a training time of 2.2 seconds.
Step two generalizes $\beta$ to the real numbers, as outlined in \cref{sec:stochastic-subspace-ppca},
and applies a one-dimensional optimization algorithm.
The search is constrained to the interval $[4, 6]$,
chosen to enclose the optimal integer value found in step one.
To reduce the approximation error in the objective function,
the number of Monte Carlo samples is increased to $10^5$.
A convergence tolerance of $10^{-10}$ is imposed to ensure numerical accuracy.
This step yields an optimal value of $\beta = 4.32$ with a training time of 577.5 seconds.
This two-step approach maintains low overall computational cost
while significantly improving predictive consistency.
Generalizing $\beta$ to real values increases the coverage of the 95\% prediction interval
from 82.2\% to 94.2\%.
It should be noted that while the hyperparameter optimization is based on the noisy observed values,
we achieve consistent coverage of the noiseless ground truth state.
In comparison, the coverage of the noisy data points is 73.68\%,
which is much lower than the coverage of the noiseless ground truth state.

\begin{figure}[!t]
  \centering
  \includegraphics[width=\textwidth]{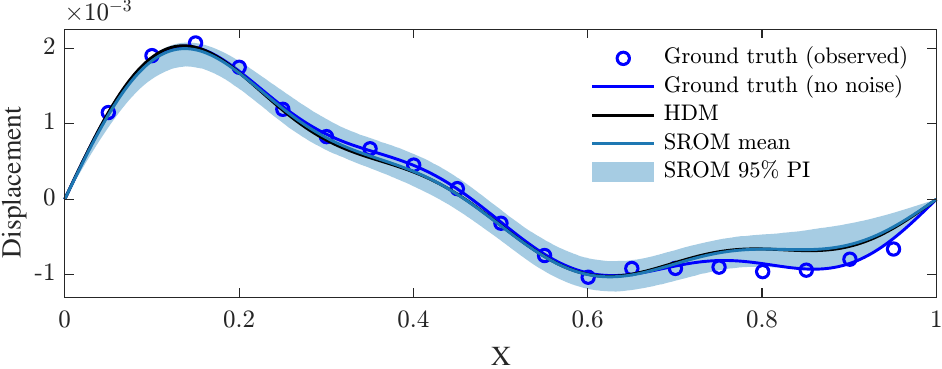}
  \caption{Linear static problem with experimental data: SS-PPCA prediction.}
  \label{fig:ex3-PPCA-prediction}
\end{figure}

\Cref{fig:ex3-PPCA-prediction} shows the HDM error characterization
by the SS-PPCA method using the sparse experimental observations.
It can be observed that our method successfully captures the model discrepancy,
yielding accurate uncertainty quantification for the HDM predictions.
The SROM 95\% PI captures the noiseless ground-truth data well, and the bounds are very sharp.
Although we do not explicitly separate measurement noise from model error during hyperparameter training, our results show that SROM coverage is higher and more consistent for noiseless data than for data with measurement noise. This suggests that SS-PPCA implicitly distinguishes between the two. While we do not provide a formal theoretical justification, this observation indicates promising behavior that merits further investigation.
The trained SROM model via SS-PPCA enables fast prediction with error estimates and eliminates the dependence on computationally expensive HDM. 

\subsection{Dynamics problem: space structure}\label{sec:dynamic}

In \cref{sec:param-static,sec:static}, we demonstrated that the SS-PPCA method
performs well in characterizing uncertainty with low computational cost.
However, those cases involved relatively simple static problems.
To evaluate the robustness and scalability of the method,
we compare it in a much more complex linear dynamics problem of a space structure.
Because of the higher model complexity and more realistic structural dynamics,
this example may better assess the efficiency and effectiveness of the method.

\begin{figure}[!t]
  \centering
  \includegraphics[width=0.9\linewidth]{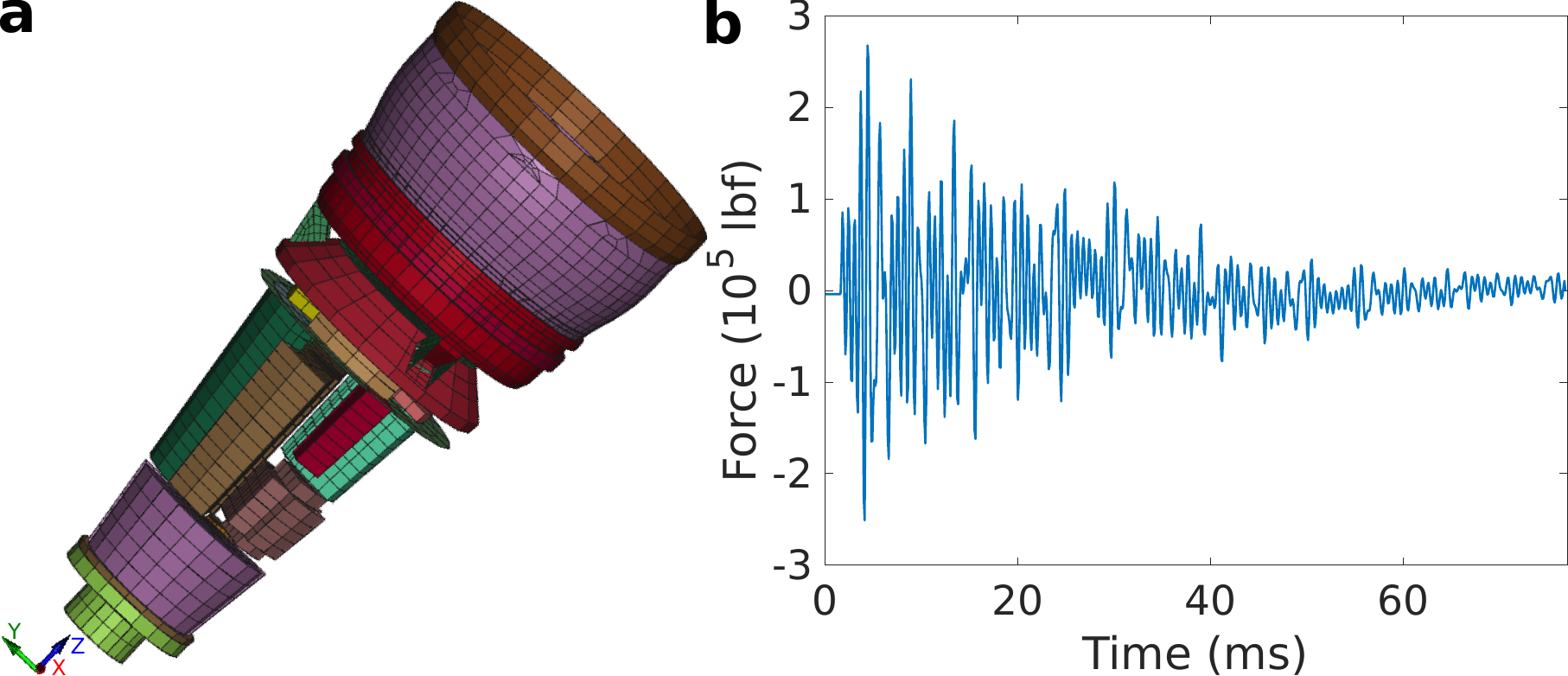}
  \caption{Dynamics problem: (a) space structure; (b) loading.}
  \label{fig:dynamics-problem}
\end{figure}

We consider a major component of a space structure \cite{ZengX2023}
given an impulse load in the $z$-direction,
see \Cref{fig:dynamics-problem}.
The space structure consists of two major parts:
an open upper and a lower part.
A large mass (approximately 100 times heavier than the remainder of the model)
is placed at the center of the upper part.
This mass is connected to the sidewalls via rigid links.
The lower part of the structure consists of the outer cylindrical shell
and an inner shock absorption block with a hollow space to hold the essential components.
The impulse load, as shown in \Cref{fig:dynamics-problem}b,
is applied at the center of the mass in the $z$-direction,
which is transmitted to the upper part through the sidewalls
and then to the lower part via the mounting pedestal.
The impulse load can compromise the functionality of the essential components located in the lower part.
Hence, monitoring the acceleration and velocity at the critical points of the essential parts
is important for their safe operation.
We take the quantity of interest (QoI)
to be the $x$-velocity of a critical point at one of the essential components.
The space structure has no boundary conditions; it can be assumed to be floating in outer space.
The governing equation of the HDM of the space structure is given by:
\begin{equation}\label{eq:HDM-space}
  \mathbf{M}_{\text{H}} \, \ddot{\mathbf{x}} + \mathbf{C}_{\text{H}} \, \dot{\mathbf{x}} + \mathbf{K}_{\text{H}} \, \mathbf{x} =
  \mathbf{f},
\end{equation}
with the initial conditions $\mathbf{x}(0) = 0$ and $\dot{\mathbf{x}}(0) = 0$.
The system matrices $\mathbf{M}_{\text{H}}, \mathbf{K}_{\text{H}}$,
and force $\mathbf{f}$ are exported from the finite element model in LS-DYNA
and $\mathbf{C}_{\text{H}} = \beta_{\text{H}} \mathbf{K}_{\text{H}}$ with a Raleigh damping coefficient
$\beta_{\text{H}} = \num{6.366E-6}$.
We denote the solution of the HDM as $\mathbf{x}_{\text{H}}$,
which is obtained by numerically integrating \cref{eq:HDM-space}
using the Newmark-$\beta$ method with a time step of $5\times 10^{-2}$ ms.
A single HDM simulation takes approximately 38 minutes, making it computationally expensive,
especially when multiple runs are needed.
To address this challenge, we construct a ROM with dimension $k = 10$ via POD.
The ROM derived from the HDM is defined as
\begin{equation}
  \mathbf{x}_{\text{R}} = \mathbf{V}\mathbf{q},
  \quad \mathbf{M}_{\text{H},\mathbf{V}} \, \ddot{\mathbf{q}} + \mathbf{C}_{\text{H},\mathbf{V}} \, \dot{\mathbf{q}} + \mathbf{K}_{\text{H},\mathbf{V}} \, \mathbf{q} = \mathbf{V}^{\intercal} \mathbf{f},
\end{equation}
where reduced matrices are denoted with a pattern $\mathbf{A}_{\mathbf{V}} := \mathbf{V}^{\intercal} \mathbf{A} \mathbf{V}$,
so that $\mathbf{M}_{\text{H},\mathbf{V}} = \mathbf{V}^{\intercal} \mathbf{M}_{\text{H}} \mathbf{V}$, for example.
The initial condition of ROM is given by $\mathbf{q}(0) = 0$ and $\dot{\mathbf{q}}(0) = 0$.
In contrast to the high simulation time of HDM,
ROM takes 0.2 seconds---roughly 11,200 times faster.
However, this computational efficiency comes at the expense of reduced accuracy.
Hence, error characterization of ROM is essential,
which is done by SROM throughout this study.
To construct the SROMs,
we substitute the deterministic ROB $\mathbf{V}$
by the SROBs $\mathbf{W}$.

For the SS-PPCA method, the SROBs $\mathbf{W}$ are sampled
using the stochastic subspace model described in \cref{sec:stochastic-subspace-ppca} and \cref{alg:ss-ppca}.
For this example, $\mathbf{u}_E$ corresponds to the velocity from the HDM at a critical structural node,
and $\mathbf{u}_{L}^o$ corresponds to the velocity predicted by ROM at the same location.
The optimum value of the hyperparameter $\beta$ comes out to be 39,
which is optimized using the strategy given in \cref{sec:training}.
The hyperparameter training time for this example is 785 seconds, which is low compared to the computational cost of the HDM.
This shows that our method is scalable and can provide fast and reliable predictions even in complex systems. 

\begin{figure}[!t]
  \centering
  \includegraphics[width=\linewidth]{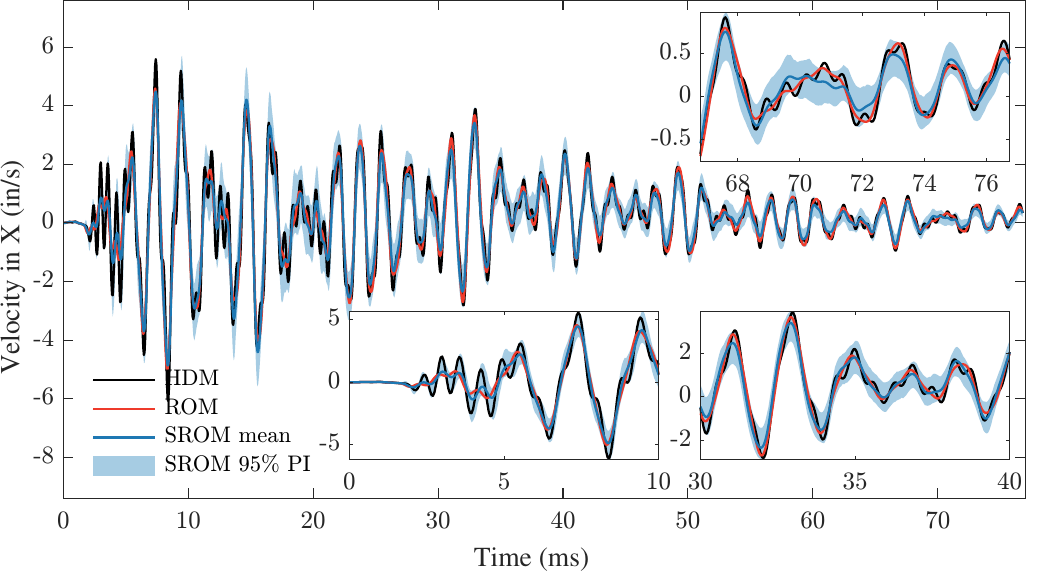}
  \caption{Dynamic prediction of the SS-PPCA model: velocity.}
  \label{fig:prediction-PPCA-vel}
\end{figure}

The HDM velocity shows a high-frequency behavior, especially during the initial 10 milliseconds.
The ROM with just 10 modes fails to capture this behavior.
\Cref{fig:prediction-PPCA-vel} illustrates that the 95\% PI of the SS-PPCA method is notably consistent and significantly sharp,
even in the initial 10 milliseconds period (see the \Cref{fig:prediction-PPCA-vel} inset for $[0, 10]$ ms). 
The SS-PPCA method not only captures the high-frequency behavior but also maintains tight bounds around it.
Additionally, the SROM mean velocity of SS-PPCA closely matches the deterministic ROM velocity, as it should be.
The coverage for velocity via the SS-PPCA method is 94.83\%.

\begin{figure}[!t]
  \centering
  \includegraphics[width=\linewidth]{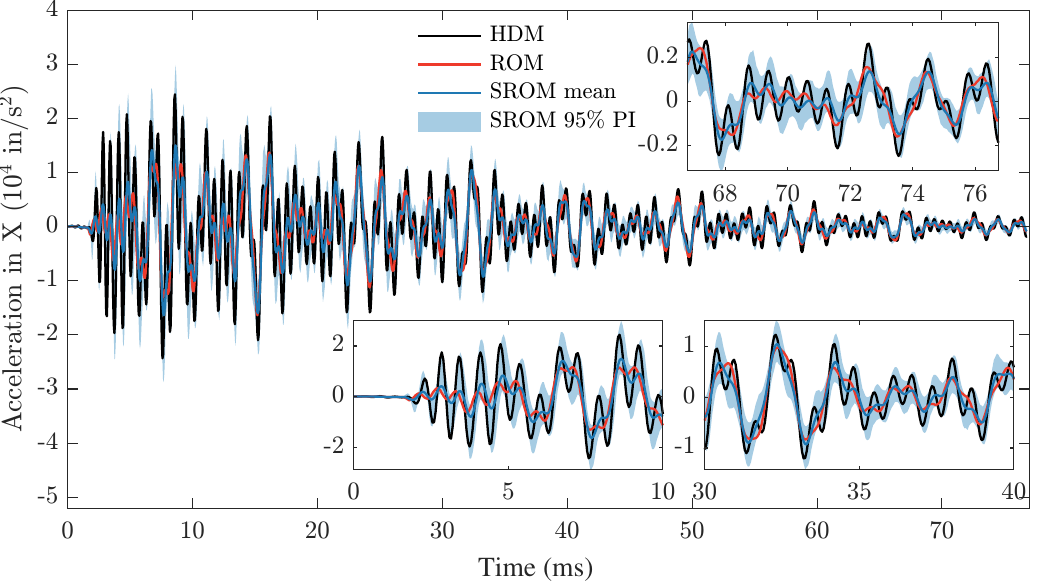}
  \caption{Dynamic prediction of the SS-PPCA model: acceleration.}
  \label{fig:prediction-PPCA-acc}
\end{figure}

\Cref{fig:prediction-PPCA-acc} shows the $x$-direction acceleration
at the critical node predicted by the SS-PPCA. %
The acceleration also exhibits a high-frequency behavior, due to the impulse loading.
Despite this, our SS-PPCA method is able to produce consistent UQ with sharp bounds. %
It is important to note that even though
velocity (the QoI for this example) is used for training the hyperparameters
and displacement is used for obtaining the deterministic ROB $\mathbf{V}$,
our approach enables predictions of unobserved quantities such as acceleration.
The fact that acceleration has not been observed makes our result very intriguing.
Furthermore, our approach enables efficient predictions of any unobserved QoIs of the whole system,
a capability that many studies in the literature fail to achieve.
This highlights the effectiveness and practicality of our approach to engineering applications.
The coverage for acceleration via the SS-PPCA method is 85.49\%, which means that the error characterization can be improved a bit. 
However, given that the acceleration is an unobserved quantity and exhibits a high-frequency behavior, the result is quite good.

\begin{figure}[!t]
  \centering
  \includegraphics[width=\linewidth]{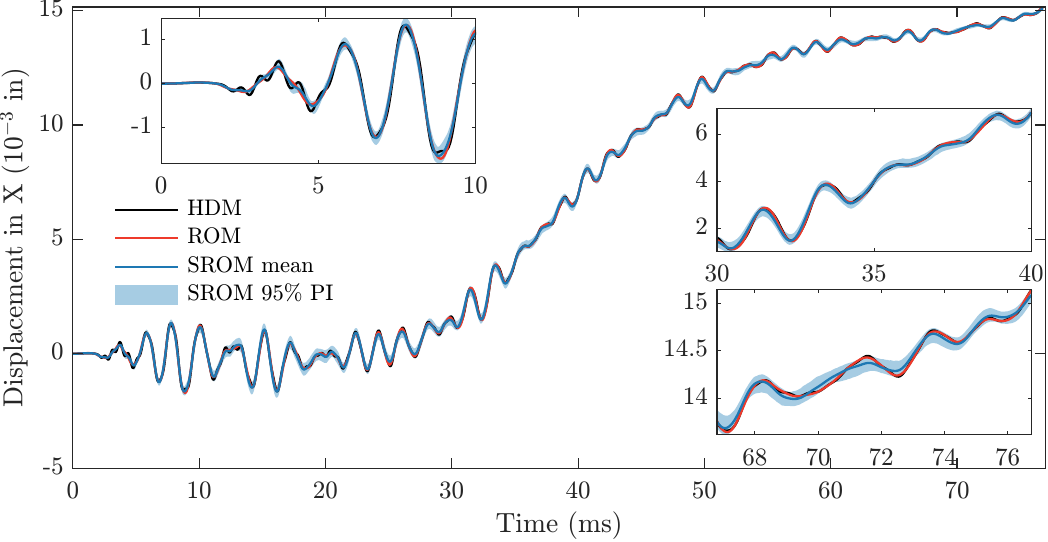}
  \caption{Dynamic prediction of the SS-PPCA model: displacement.}
  \label{fig:prediction-PPCA-disp}
\end{figure}

\Cref{fig:prediction-PPCA-disp} compares the $x$-direction displacement
at the critical node predicted by the SS-PPCA.
It can be observed that SS-PPCA displacement prediction is highly accurate:
the SROM, ROM, and HDM show similar behavior, and the 95\% PI is very tight.
The coverage for displacement via the SS-PPCA method is 96.78 \%, which means that the displacement prediction is consistent and highly accurate.
\begin{figure}[!t]
  \centering
  \includegraphics[width=\linewidth]{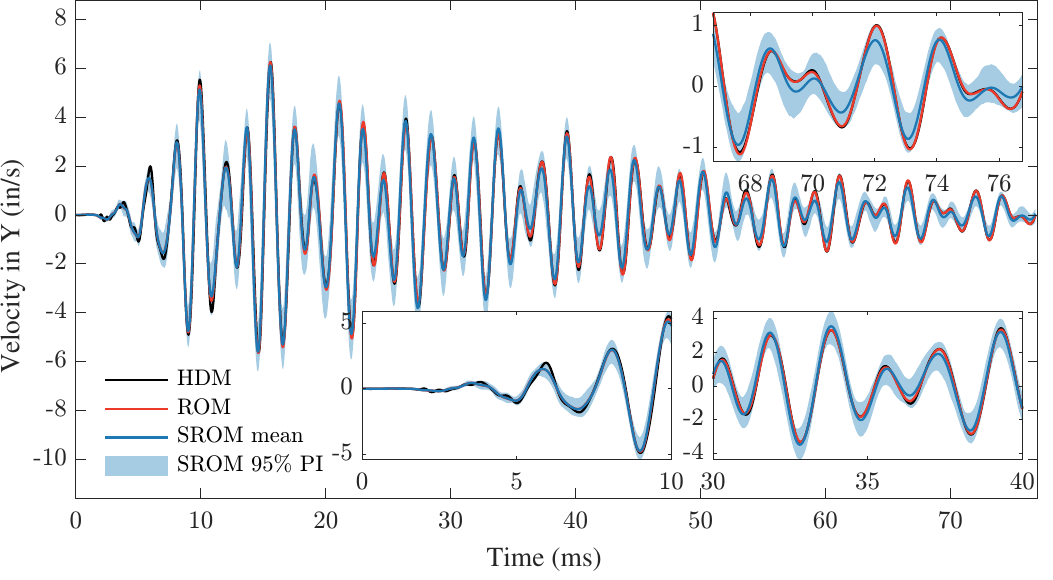}
  \caption{Dynamic prediction of the SS-PPCA model: velocity at a random DoF.}
  \label{fig:prediction-PPCA-vel-rand}
\end{figure}

All the quantities compared in \Cref{fig:prediction-PPCA-vel,fig:prediction-PPCA-acc,fig:prediction-PPCA-disp}, whether observed or unobserved, are from the same node of the space structure.
However, we may need to make a prediction at a different node of the system,
or we may be interested in the general behavior of the system.
\Cref{fig:prediction-PPCA-vel-rand} shows the prediction of the SS-PPCA method
for the $y$-direction velocity at a random node; the result is still consistent and very sharp.
This shows that the SS-PPCA methods allow us to make predictions about the entire structure
just by training on the observation data of a single DoF.

In summary, the low computational cost of hyperparameter optimization of SS-PPCA, combined with the low computational cost of ROM,
can accelerate the prediction of complex engineering phenomena.
In particular, the hyperparameter optimization of SS-PPCA is far less expensive than that of the NPM method, whose large number of hyperparameters leads to significantly higher computational cost, as discussed in \cref{sec:NPM}.
Furthermore, SS-PPCA offers a fully automated, efficient, and principled training workflow with minimal user input.
This makes it well suited for practical applications.
One issue with methods based on SROM
is that they cannot fully eliminate model error as they rely on a ROM.
This issue can be addressed by increasing the number of reduced bases or incorporating model closure terms.
The reduced-order dimension $k$ can be selected using the standard criteria mentioned in \cref{sec:ROM}. 
Increasing $k$ can reduce error due to mode truncation, but it may also increase computational cost.
Thus, the choice of $k$ requires a trade-off between accuracy and efficiency.
A study on model closure for ROMs is presented by \citet{Ahmed2021}.

\section{Conclusion} \label{sec:conclusion}

We introduced a novel stochastic subspace model,
which is used to characterize model error in the framework of stochastic reduced-order models.
It is simple, efficient, easy to implement, and well-supported by analytical results.
Our SS-PPCA model has only one hyperparameter, which is optimized systematically to improve consistency.
Through various numerical examples, we reveal the characteristics of this model,
show its flexibility across single and multiple model cases,
establish its consistency, and quantify its remarkable sharpness.
It opens up a promising path to the challenging problem of model-form uncertainty.

Across all examples, the stochastic subspace model effectively captures various forms of model error,
including parametric uncertainty, ROM error, and HDM errors informed by noisy experimental data.
The examples include both linear and nonlinear systems, demonstrating the broad applicability of the proposed method for characterizing model error across different system complexities.
However, in scenarios with large errors, mere characterization may not suffice.
In such cases, future work will focus on developing practical and efficient techniques for model error correction.

\section*{Acknowledgments}

The authors thank Johann Guilleminot and Christian Soize for invaluable discussions on this topic.
This work was supported by the University of Houston.

\clearpage
\appendix
\section{Mode of $\texorpdfstring{\text{MACG}_{n,k}(\boldsymbol{\Sigma})}{MACG_{n,k}(\Sigma)}$}
\label{apd:mode-proof}

In this section, we establish the mode of the MACG probability distribution on Grassmann manifolds.
To the best of our knowledge, this is the first appearance of such results.

The matrix angular central Gaussian distribution $\text{MACG}_{n,k}(\boldsymbol{\Sigma})$
on the Grassmann manifold $\text{Gr}(n, k)$ \cite{Chikuse2003, Chikuse1990}
is the probability distribution with the probability density function (PDF)
$p(\mathbf{X} \mathbf{X}^\intercal; \boldsymbol{\Sigma}) = z^{-1} |\mathbf{X}^\intercal \boldsymbol{\Sigma}^{-1} \mathbf{X}|^{-n/2}$
with respect to the normalized invariant measure of $\text{Gr}(n, k)$.
Here $\boldsymbol{\Sigma} > 0$ is an order-$n$ symmetric positive definite matrix,
a $k$-dim subspace $\mathscr{X} = \text{range}(\mathbf{X})$ is represented by
the orthogonal projection matrix $\mathbf{X} \mathbf{X}^\intercal$
where $\mathbf{X} \in \text{St}(n, k)$ is an orthonormal basis,
and the normalizing constant $z = |\boldsymbol{\Sigma}|^{k/2}$
where $|\cdot|$ denotes the matrix determinant.
In the following, let $\boldsymbol{\Sigma}$ have eigenvalues $\lambda_1 \ge \cdots \ge \lambda_n > 0$
and corresponding eigenvectors $\mathbf{v}_1, \cdots, \mathbf{v}_n$ that form an orthonormal basis of $\mathbb{R}^n$.

\begin{theorem}
  The PDF of the $\text{MACG}_{n,k}(\boldsymbol{\Sigma})$ distribution on $\text{Gr}(n, k)$ is maximized
  at any $k$-dim principal subspace of $\boldsymbol{\Sigma}$.
  If $\lambda_k > \lambda_{k+1}$, then the principal subspace $\mathscr{V}_k = \text{range}(\mathbf{V}_k)$
  with $\mathbf{V}_k = [\mathbf{v}_1 \cdots \mathbf{v}_k]$ is the unique mode of the distribution.
\end{theorem}

\begin{proof}
  Let us first consider a special case when $\boldsymbol{\Sigma}$ is the diagonal matrix
  $\boldsymbol{\Lambda} = \diag(\lambda_1, \cdots, \lambda_n)$.
  Let subspace $\mathscr{X} \sim \text{MACG}_{n,k}(\boldsymbol{\Lambda})$
  and have an orthonormal basis $\mathbf{X} \in \text{St}(n,k)$.
  The PDF can be written as:
  \begin{equation}
    p(\mathscr{X}) =
    |\boldsymbol{\Lambda}|^{-k/2} |\mathbf{X}^\intercal \boldsymbol{\Lambda}^{-1}\mathbf{X}|^{-n/2} =
    \left(\prod_{i=1}^n \lambda_i \right)^{-k/2} |\mathbf{X}^\intercal \boldsymbol{\Lambda}^{-1}\mathbf{X}|^{-n/2}.
  \end{equation}
  We see that the PDF is a smooth function of $\mathbf{X}$, and we first characterize its critical points.
  Taking the logarithm of the above equation, we get:
  \begin{equation}\label{eq:objective-function-mode-MACG}
    \log p(\mathscr{X}) = -\frac{n}{2} \log |\mathbf{X}^\intercal \boldsymbol{\Lambda}^{-1}\mathbf{X}| - \frac{k}{2} \sum_{i=1}^{n} \log \lambda_i.
  \end{equation}
  The second term $\frac{k}{2} \sum_{i=1}^{n} \log \lambda_i$ is a constant,
  so the critical points of $p(\mathscr{X})$ are the critical points of
  $\log |\mathbf{X}^\intercal \boldsymbol{\Lambda}^{-1} \mathbf{X}|$,
  which can be computed by setting $ \tr(\mathbf{S}^{-1} \partial \mathbf{S}) = 0$
  where $\mathbf{S} = \mathbf{X}^\intercal \boldsymbol{\Lambda}^{-1} \mathbf{X}$.
  The partial derivative of $\mathbf{X}^\intercal \boldsymbol{\Lambda}^{-1} \mathbf{X}$ can be written as
  $\partial (\mathbf{X}^\intercal \boldsymbol{\Lambda}^{-1} \mathbf{X}) = (\partial \mathbf{X})^\intercal \boldsymbol{\Lambda}^{-1} \mathbf{X} + \mathbf{X}^\intercal \boldsymbol{\Lambda}^{-1} (\partial \mathbf{X})$.
  The differential $\partial \mathbf{X} \in T_{\mathbf{X}} \text{St}(n,k)$,
  meaning that it lies in the tangent space of the Stifel manifold $\text{St}(n,k)$ at the point $\mathbf{X}$,
  which must satisfy $(\partial \mathbf{X})^\intercal \mathbf{X} + \mathbf{X}^\intercal (\partial \mathbf{X}) = 0$.
  In other words, $(\partial \mathbf{X})^\intercal \mathbf{X}$ must be antisymmetric,
  which allows us to write $\partial \mathbf{X}$ in a projective form:
  $\partial \mathbf{X} = \mathbf{M} - \mathbf{X} \sym(\mathbf{X}^\intercal \mathbf{M})$
  where $\mathbf{M} \in \mathbb{R}^{n \times k}$ and $\text{sym}(\mathbf{A}) = (\mathbf{A} + \mathbf{A}^\intercal) / 2$.
  Now we have:
  \begin{equation}
    \begin{aligned}
      \tr(\mathbf{S}^{-1} \partial \mathbf{S})
      &= \tr \{ (\mathbf{X}^\intercal \boldsymbol{\Lambda}^{-1} \mathbf{X})^{-1} [(\partial \mathbf{X})^\intercal \boldsymbol{\Lambda}^{-1} \mathbf{X} + \mathbf{X}^\intercal \boldsymbol{\Lambda}^{-1} (\partial \mathbf{X})]\}\\
      &= \tr\{(\partial \mathbf{X})^\intercal \boldsymbol{\Lambda}^{-1} \mathbf{X} (\mathbf{X}^\intercal \boldsymbol{\Lambda}^{-1}\mathbf{X})^{-1} + (\mathbf{X}^\intercal \boldsymbol{\Lambda}^{-1}\mathbf{X})^{-1}\mathbf{X}^\intercal \boldsymbol{\Lambda}^{-1}(\partial \mathbf{X})\}\\
      &= 2\tr\{(\mathbf{X}^\intercal \boldsymbol{\Lambda}^{-1}\mathbf{X})^{-1}\mathbf{X}^\intercal \boldsymbol{\Lambda}^{-1}(\partial \mathbf{X})\}\\
      &= 2\tr\{(\mathbf{X}^\intercal \boldsymbol{\Lambda}^{-1}\mathbf{X})^{-1}\mathbf{X}^\intercal \boldsymbol{\Lambda}^{-1}[\mathbf{M} - \mathbf{X} \sym(\mathbf{X}^\intercal \mathbf{M})]\}\\
      &= 2\tr\{(\mathbf{X}^\intercal \boldsymbol{\Lambda}^{-1}\mathbf{X})^{-1}\mathbf{X}^\intercal \boldsymbol{\Lambda}^{-1}\mathbf{M} - \sym(\mathbf{X}^\intercal \mathbf{M})\}\\
      &= 2\tr\{(\mathbf{X}^\intercal \boldsymbol{\Lambda}^{-1}\mathbf{X})^{-1}\mathbf{X}^\intercal \boldsymbol{\Lambda}^{-1}\mathbf{M} - \mathbf{X}^\intercal \mathbf{M}\}\\
      &= 2 \, \text{sum}\{[(\mathbf{X}^\intercal \boldsymbol{\Lambda}^{-1}\mathbf{X})^{-1}\mathbf{X}^\intercal \boldsymbol{\Lambda}^{-1} - \mathbf{X}^\intercal] \circ \mathbf{M}\}.
    \end{aligned}
  \end{equation}
  With $\tr(\mathbf{S}^{-1} \partial \mathbf{S}) = 0$ for all $\mathbf{M} \in \mathbb{R}^{n \times k}$,
  we have $(\mathbf{X}^\intercal \boldsymbol{\Lambda}^{-1}\mathbf{X})^{-1}\mathbf{X}^\intercal \boldsymbol{\Lambda}^{-1} - \mathbf{X}^\intercal = 0$,
  that is, $\mathbf{X} = \boldsymbol{\Lambda}^{-1}\mathbf{X}(\mathbf{X}^\intercal \boldsymbol{\Lambda}^{-1}\mathbf{X})^{-1}$.
  Further simplification gives us
  $(\mathbf{X}\mathbf{X}^\intercal)(\boldsymbol{\Lambda}^{-1}\mathbf{X}) = (\boldsymbol{\Lambda}^{-1}\mathbf{X})$.
  Because $\mathbf{X}\mathbf{X}^\intercal$ is the orthogonal projection onto $\text{range}(\mathbf{X})$,
  we have $\text{range}(\boldsymbol{\Lambda}^{-1}\mathbf{X}) \subseteq \text{range}(\mathbf{X})$;
  since both sides of the equation are $k$-dim subspaces,
  we have $\text{range}(\boldsymbol{\Lambda}^{-1}\mathbf{X}) = \text{range}(\mathbf{X})$.
  Therefore, $\mathscr{X} = \text{range}(\mathbf{X})$ is an invariant $k$-subspace of $\boldsymbol{\Lambda}^{-1}$,
  or equivalently, an invariant $k$-subspace of $\boldsymbol{\Lambda}$.
  If the eigenvalues $(\lambda_i)_{i=1}^{n}$ are distinct,
  then $\mathscr{X}$ must contain a $k$-combination of the standard basis $\{\mathbf{e}_i\}_{i=1}^{n}$.

  So far we have proved that the critical points of $p(\mathscr{X})$ are the invariant $k$-subspaces of $\boldsymbol{\Lambda}$.
  Since the Grassmann manifold $\text{Gr}(n, k)$ is compact and without a boundary,
  the smooth PDF $p(\mathscr{X})$ has a global minimum and a global maximum,
  both of which are critical points.
  We define two $k$-dim subspaces $\overline{\mathscr{E}}_k$ and $\underline{\mathscr{E}}_k$
  as the ranges of the bases $\overline{\mathbf{E}}_k := [\mathbf{e}_1 \cdots \mathbf{e}_k]$
  and $\underline{\mathbf{E}}_k := [\mathbf{e}_{n - k + 1} \cdots \mathbf{e}_n]$, respectively.
  We have:
  \begin{equation}
    \begin{aligned}
      \log p(\overline{\mathscr{E}}_k)
      &= -\frac{n}{2} \log |\overline{\mathbf{E}}_{k}^\intercal \boldsymbol{\Lambda}^{-1} \overline{\mathbf{E}}_k| - \frac{k}{2} \sum_{i=1}^{n} \log \lambda_i \\
      &= -\frac{n}{2} \log |\diag(\lambda_1^{-1}, \cdots, \lambda_k^{-1})| - \frac{k}{2} \sum_{i=1}^{n} \log \lambda_i \\
      &= -\frac{n}{2} \sum_{i=1}^{k} \log \lambda_{i}^{-1} - \frac{k}{2} \sum_{i=1}^{n} \log \lambda_i \\
      &= \frac{n}{2} \sum_{i=1}^{k} \log \lambda_{i} - \frac{k}{2} \sum_{i=1}^{n} \log \lambda_i.
    \end{aligned}
  \end{equation}
  Similarly, we have $\log p(\underline{\mathscr{E}}_k) = \frac{n}{2} \sum_{i=n-k+1}^{n} \log \lambda_{i} - \frac{k}{2} \sum_{i=1}^{n} \log \lambda_i$.
  In general, let $\lambda_{a_1} \ge \cdots \ge \lambda_{a_k}$ be the eigenvalues associated with
  an invariant $k$-subspace $\mathscr{X}_0$, then:
  \begin{equation}
    \log p(\mathscr{X}_0) = \frac{n}{2} \sum_{j=1}^{k} \log \lambda_{a_j} - \frac{k}{2} \sum_{i=1}^{n} \log \lambda_i.
  \end{equation}
  Since $\lambda_1 \ge \cdots \ge \lambda_n$, we have
  $p(\underline{\mathscr{E}}_k) \leq p(\mathscr{X}_0) \leq p(\overline{\mathscr{E}}_k)$,
  for all critical points $\mathscr{X}_0$.
  This means that $\underline{\mathscr{E}}_k$ and $\overline{\mathscr{E}}_k$
  are the global minimal point and the global maximal point of $p(\mathscr{X})$, respectively;
  that is, $p(\underline{\mathscr{E}}_k) \leq p(\mathscr{X}) \leq p(\overline{\mathscr{E}}_k)$.

  Now we show that a critical point $\mathscr{X}_0$ cannot be a local maximum unless
  $\{\lambda_{a_1}, \cdots, \lambda_{a_k}\} = \{\lambda_1, \cdots, \lambda_k\}$.
  Let $\mathbf{v}_{a_1}, \cdots, \mathbf{v}_{a_k}$ be eigenvectors associated with eigenvalues
  $\lambda_{a_1}, \cdots, \lambda_{a_k}$, respectively,
  such that they form an orthonormal basis of $\mathscr{X}_0$.
  Assume that eigenvalue $\lambda' \notin \{\lambda_{a_1}, \cdots, \lambda_{a_k}\}$ and $\lambda' > \lambda_{a_k}$.
  Let $\mathbf{v}'$ be an eigenvector associated with the eigenvalue $\lambda'$.
  Consider the trajectory $\mathscr{X}(\theta) = \text{range}(\mathbf{X}(\theta))$,
  where $\mathbf{X}(\theta) = [\mathbf{v}_{a_1} \, \cdots \, \mathbf{v}_{a_{k-1}} \, \mathbf{v}(\theta)]$,
  $\mathbf{v}(\theta) = \cos(\theta) \mathbf{v}_{a_k} + \sin(\theta) \mathbf{v}'$,
  and $\theta \in [0, \frac{\pi}{2}]$.
  It is a constant-speed rotation that starts from $\mathscr{X}_0$.
  We can show that:
  \begin{equation}
    \log p(\mathscr{X}(\theta)) = \frac{n}{2} \left( \sum_{j=1}^{k-1} \log \lambda_{a_j}
      - \log \left( \lambda_{a_k}^{-1} \cos^2\theta + \lambda'^{-1} \sin^2\theta \right) \right)
    - \frac{k}{2} \sum_{i=1}^{n} \log \lambda_i,
  \end{equation}
  which increases monotonically.
  Therefore, $\mathscr{X}_0$ is not a local maximum.

  In summary, $\text{MACG}_{n,k}(\boldsymbol{\Lambda})$ has its global maximum at $\overline{\mathscr{E}}_k$,
  or any $k$-dim principal subspace of $\boldsymbol{\Lambda}$.
  These subspaces are also the modes of the distribution.
  If $\lambda_k > \lambda_{k+1}$, there is a unique $k$-dim principal subspace $\overline{\mathscr{E}}_k$,
  which is the unique mode of the distribution.

  For a general $\boldsymbol{\Sigma}$, it has an eigenvalue decomposition
  $\boldsymbol{\Sigma} = \mathbf{V} \boldsymbol{\Lambda} \mathbf{V}^\intercal$
  where $\mathbf{V} = [\mathbf{v}_1 \, \cdots \mathbf{v}_k]$.
  Given a change of basis $\mathbf{x} = \mathbf{V} \mathbf{z}$, the previous arguments still hold,
  and the theorem follows immediately.
\end{proof}

\bibliographystyle{myelsarticle-num-names}
{\bibliography{SROM}}

\end{document}